\documentclass[11pt,letterpaper]{article}
\usepackage[lmargin=1.0in,rmargin=1.0in,bottom=1.0in,top=1.0in,twoside=False]{geometry}

\usepackage{fullpage,amssymb,amsmath}
\usepackage{graphicx}
\usepackage{tikz}
\usetikzlibrary{shapes}

\usepackage[T1]{fontenc}

 \usepackage{xcolor}
 \usepackage{mathtools}
\usepackage{microtype}
\usepackage{amsfonts}
\usepackage{comment}
\usepackage[english]{babel}
\usepackage{mathrsfs}
\usepackage[vlined, ruled, linesnumbered]{algorithm2e}
\DontPrintSemicolon

\usepackage{trimspaces}
\usepackage{nccfoots}
\usepackage{setspace}
\usepackage{inconsolata}
\usepackage{libertine}
\usepackage[absolute]{textpos}
\usepackage{thmtools}
\usepackage{thm-restate}

\usepackage{todonotes}

\usepackage{longtable}

\definecolor{blue}{rgb}{0.1,0.2,0.5}
\definecolor{brown}{rgb}{0.6,0.6,0.2}
\usepackage[ocgcolorlinks, linkcolor={blue}, citecolor={brown}]{hyperref}

\usepackage{comment}

\usepackage[amsmath,thmmarks,hyperref]{ntheorem}
\usepackage{cleveref}

\crefformat{page}{#2page~#1#3}%
\Crefformat{page}{#2Page~#1#3}%
\crefformat{equation}{#2(#1)#3}%
\Crefformat{equation}{#2(#1)#3}%
\crefformat{figure}{#2Figure~#1#3}%
\Crefformat{figure}{#2Figure~#1#3}%
\crefformat{section}{#2Section~#1#3}
\Crefformat{section}{#2Section~#1#3}
\crefformat{chapter}{#2Chapter~#1#3}
\Crefformat{chapter}{#2Chapter~#1#3}
\crefformat{chapter*}{#2Chapter~#1#3}
\Crefformat{chapter*}{#2Chapter~#1#3}
\crefformat{part}{#2Part~#1#3}
\Crefformat{part}{#2Part~#1#3}
\crefformat{enumi}{#2(#1)#3}
\Crefformat{enumi}{#2(#1)#3}

\usepackage{enumerate}

\usepackage{latexsym}

\theoremnumbering{arabic}
\theoremstyle{plain}
\theoremsymbol{}
\theorembodyfont{\itshape}
\theoremheaderfont{\normalfont\bfseries}
\theoremseparator{.}

\newtheorem{theorem}{Theorem}
\crefformat{theorem}{#2Theorem~#1#3}
\Crefformat{theorem}{#2Theorem~#1#3}

\newcommand{\newtheoremwithcrefformat}[2]{%
  \newtheorem{#1}[theorem]{#2}%
  \crefformat{#1}{##2\MakeUppercase#1~##1##3}%
  \Crefformat{#1}{##2\MakeUppercase#1~##1##3}%
}
\newcommand{\newseptheoremwithcrefformat}[2]{%
  \newtheorem{#1}{#2}%
  \crefformat{#1}{##2\MakeUppercase#1~##1##3}%
  \Crefformat{#1}{##2\MakeUppercase#1~##1##3}%
}

\newtheoremwithcrefformat{lemma}{Lemma}
\newtheoremwithcrefformat{proposition}{Proposition}
\newtheoremwithcrefformat{observation}{Observation}
\newtheoremwithcrefformat{conjecture}{Conjecture}
\newtheoremwithcrefformat{corollary}{Corollary}
\newseptheoremwithcrefformat{claim}{Claim}
\theorembodyfont{\upshape}
\newtheoremwithcrefformat{example}{Example}
\newtheoremwithcrefformat{remark}{Remark}
\newseptheoremwithcrefformat{definition}{Definition}
\newseptheoremwithcrefformat{question}{Question}

\theoremstyle{nonumberplain}
\theoremheaderfont{\scshape}
\theorembodyfont{\normalfont}
\theoremsymbol{\ensuremath{\square}}
\newtheorem{proof}{Proof}

\theoremsymbol{\ensuremath{\lrcorner}}

\def\cqedsymbol{\ifmmode$\lrcorner$\else{\unskip\nobreak\hfil
\penalty50\hskip1em\null\nobreak\hfil$\lrcorner$
\parfillskip=0pt\finalhyphendemerits=0\endgraf}\fi}

\tikzset{
    position/.style args={#1:#2 from #3}{
        at=(#3.#1), anchor=#1+180, shift=(#1:#2)
    }
}

\newcommand{\Bc}{\mathcal{B}}

\newcommand{\Oh}{\mathcal{O}}
\newcommand{\eps}{\varepsilon}

\let\originalleft\left
\let\originalright\right
\renewcommand{\left}{\mathopen{}\mathclose\bgroup\originalleft}
\renewcommand{\right}{\aftergroup\egroup\originalright}

\renewcommand{\leq}{\leqslant}
\renewcommand{\geq}{\geqslant}

\newcommand{\MIS}{\textsc{Maximum Weight Independent Set}\xspace}
\newcommand{\MIM}{\textsc{Maximum Weight Induced Matching}\xspace}
\newcommand{\Col}{\textsc{3-Coloring}\xspace}

\newcommand{\wei}{\mathfrak{w}}

\usepackage{todonotes}

\usepackage{lineno}

\begin{document}

\title{\Large Quasi-polynomial-time algorithm for Independent Set in $P_t$-free graphs via shrinking the space of induced paths}

\author{
Marcin{}~Pilipczuk\thanks{
  Institute of Informatics, University of Warsaw, Poland, \texttt{malcin@mimuw.edu.pl}.
  This work is 
a part of project CUTACOMBS that has received funding from the European Research Council (ERC) 
under the European Union's Horizon 2020 research and innovation programme (grant agreement No.~714704).
}
\and
Micha\l{}~Pilipczuk\thanks{
  Institute of Informatics, University of Warsaw, Poland, \texttt{michal.pilipczuk@mimuw.edu.pl}.
  This work is 
a part of project TOTAL that has received funding from the European Research Council (ERC) 
under the European Union's Horizon 2020 research and innovation programme (grant agreement No.~677651).
}
\and
Pawe\l{} Rz\k{a}\.{z}ewski\thanks{
  Faculty of Mathematics and Information Science, Warsaw University of Technology, Poland, and Institute of
Informatics, University of Warsaw, Poland, \texttt{p.rzazewski@mini.pw.edu.pl}.
Supported by Polish National Science Centre grant no. 2018/31/D/ST6/00062.
}
}

\date{}

\begin{titlepage}
\def\thepage{}
\thispagestyle{empty}
\maketitle

\begin{textblock}{20}(0, 11.7)
\includegraphics[width=40px]{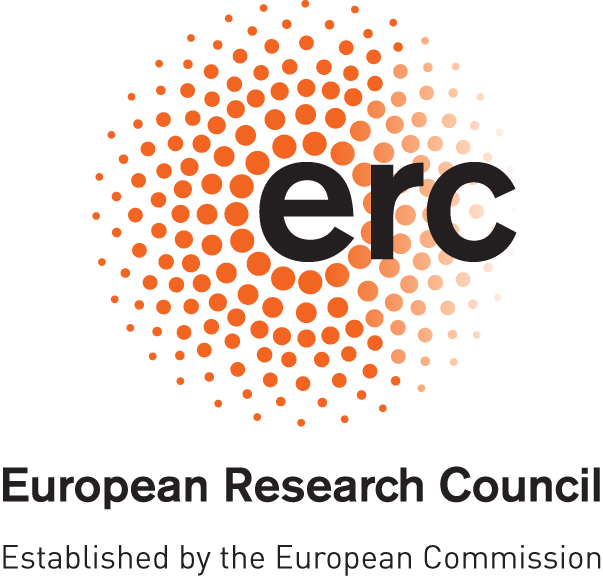}%
\end{textblock}
\begin{textblock}{20}(-0.25, 12.1)
\includegraphics[width=60px]{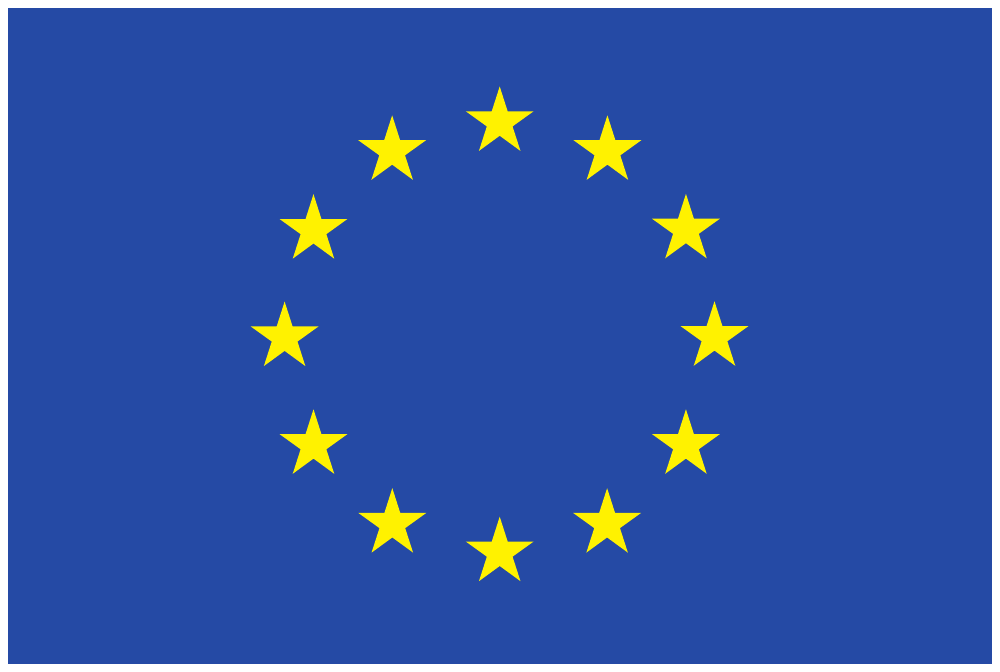}%
\end{textblock}
\begin{abstract}
In a recent breakthrough work, Gartland and Lokshtanov [FOCS 2020] showed
a quasi-polynomial-time algorithm for \textsc{Maximum Weight Independent Set} in $P_t$-free graphs, that is, 
graphs excluding a fixed path as an induced subgraph. Their algorithm runs in time $n^{\Oh(\log^3 n)}$, where $t$ is assumed to be a constant.

Inspired by their ideas, we present an arguably simpler algorithm
with an improved running time bound of $n^{\Oh(\log^2 n)}$. 
Our main insight is that a connected $P_t$-free graph always contains a vertex $w$ whose neighborhood intersects,
for a constant fraction of pairs $\{u,v\} \in \binom{V(G)}{2}$, a constant fraction of induced $u-v$ paths. 
Since a $P_t$-free graph contains $\Oh(n^{t-1})$ induced paths in total, branching on  such a vertex and recursing
independently on the connected components leads to a quasi-polynomial running time bound.

We also show that the same approach can be used to obtain quasi-polynomial-time algorithms for related problems,
including \textsc{Maximum Weight Induced Matching} and \textsc{3-Coloring}.

\end{abstract}

\end{titlepage}

\section{Introduction}
Understanding the boundary of tractability of fundamental graph problems, depending on restrictions put on the input graph,
has been a very active area of research in the last three decades. 
A methodological way of studying restrictions on the input graph is to focus on \emph{hereditary} graph classes, that is,
graph classes closed under vertex deletion. 
From this point of view, one starts with studying {\em{$H$-free graphs}}: graphs excluding one fixed graph $H$ as an induced subgraph. The goal is to classify for which graphs $H$, the problem at hand admits an efficient algorithm in $H$-free graphs.

Arguably, one of the most intriguing cases is $H =P_t$, where $P_t$ denotes a path on $t$ vertices. 
Alekseev~\cite{Alekseev82}  observed already in 1982 that for \textsc{Maximum Independent Set}, the known NP-hardness reductions
prove hardness for $H$-free graphs unless every component of $H$ is a tree with at most three leaves. 
Since then, we know no more hardness results; in particular, no NP-hardness reduction for \textsc{Maximum Independent Set}
is known that would not produce long induced paths in the output graph. 
On the positive side, polynomial time algorithms are known for only a few small graphs $H$~\cite{CorneilLB81,LokshtanovVV14,GrzesikKPP19,Sbihi80,Minty80,Al04,LozinM08}.

Very recently, we have learned that there is a reason for that. 
First, two authors of this paper together with Chudnovsky and Thomass\'{e}~\cite{ChudnovskyPPT20}
showed that \MIS
(the weighted generalization of \textsc{Maximum Independent Set}) admits a quasi-polynomial-time approximation scheme (QPTAS)
in $H$-free graphs in all cases left open by Alekseev. 
Second, in a recent breakthrough result, Gartland and Lokshtanov~\cite{GL20}
showed a quasi-polynomial-time algorithm for \MIS in $P_t$-free graphs
with running time bound $n^{\Oh(\log^3 n)}$, providing the first decisive evidence against NP-hardness of the problem in these
graph classes.

In this work, inspired by the combinatorial insights of~\cite{GL20}, we present an arguably simpler
algorithm for \MIS in $P_t$-free graphs with an improved running time bound.

\begin{restatable}{theorem}{mainthmptfree}
\label{thm:main-ptfree}
For every fixed $t\in \mathbb{N}$, the \MIS problem can be solved in time $n^{\Oh(\log^2 n)}$ in $n$-vertex $P_t$-free graphs.
\end{restatable}

We also note that both the approach of Gartland and Lokshtanov~\cite{GL20} and our approach are quite robust and can be used to obtain quasi-polynomial-time algorithms for many other graph problems. 
A notable example is \Col, whose complexity in $P_t$-free graphs is a well-known open problem~\cite{BonomoCMSSZ18,DBLP:journals/corr/abs-2006-03009}.

\begin{restatable}{theorem}{threecol}
\label{thm:threecol}
For every fixed $t\in \mathbb{N}$, the \Col problem can be solved in time $n^{\Oh(\log^2 n)}$ in $n$-vertex $P_t$-free graphs.
\end{restatable}

Let us now discuss the main ideas behind \cref{thm:main-ptfree}. We consider the space of all induced paths in the graph.
In a $P_t$-free graph, this space is small, of size $\Oh(n^{t-1})$, and thus can be enumerated in polynomial time.
The main idea is to use the size of this space to guide a branching algorithm.

For \MIS, a natural branching step is the following: take a vertex $w$ and branch into two cases: 
either take $w$ to the constructed independent set and delete $N[w]$ from the graph (the \emph{successful branch})
or do not take $w$ and delete $w$ from the graph (the \emph{failure branch}). 
In the failure branch, we cannot hope for much progress, as we delete only one vertex. 
However, if the branching pivot $w$ is chosen carefully, we can hope to guarantee large progress in the successful branch,
leading to a good running time guarantee.

The crucial combinatorial property of $P_t$-free graphs, used also in~\cite{ChudnovskyPPT20,GL20}, is the following corollary of the Gy\'arf\'as path argument.
\begin{theorem}[Gy\'arf\'as~\cite{gyarfas}]\label{lem:separator}
Let $G$ be a $P_t$-free graph.
Then there is a connected set $X$ of size at most $t$ such that every connected component of $G - N[X]$ contains at most $|V(G)|/2$ vertices.
Furthermore, such a set can be found in polynomial time.
\end{theorem}

Let $G$ be a connected $P_t$-free graph. Take all the $\Oh(n^{t-1})$ induced paths in $G$ and partition them into \emph{buckets} $\{\,\Bc_{u,v}\ \colon\ \{u,v\} \in \binom{V(G)}{2}\,\}$ according to their endpoints: $\Bc_{u,v}$ comprises all induced paths with endpoints $u$ and $v$.
Take a set $X$ given by~\cref{lem:separator}. By the separation property, $N[X]$ intersects \emph{all} paths from at least
half of all the $\binom{|V(G)|}{2}$ buckets. Hence, as $|X| \leq t$, there exists 
a vertex $w \in X$ such that $N[w]$ intersects at least $\frac{1}{t}$ paths from at least $\frac{1}{2t}$ buckets. 

Such $w$ is an excellent branching pivot: after $\Oh(\log |V(G)|)$ successful branches, a constant fraction of buckets
become empty and this implies that $G$ got disconnected into connected components of multiplicatively smaller size. 
Furthermore, since we can enumerate all buckets in polynomial time, such a vertex $w$ can be identified in polynomial time.

\medskip
In \cref{sec:ptfree} we prove formally \cref{thm:main-ptfree}.
In \cref{sec:cor} we discuss the possible extenstions of the algorithm for \Col and related problems. These extensions follow by suitably adapting the strategy outlined above.

\section{The algorithm}\label{sec:ptfree}
In this section we prove \cref{thm:main-ptfree}, restated below.

\mainthmptfree*

Let $(G,\wei)$ be the instance of the \MIS problem, where $G$ is $P_t$-free and $\wei$ is the weight function.
Without loss of generality, assume $t \geq 5$.
To simplify the notation, we allow that the domain of $\wei$ to be a superset of $V(G)$.

Consider the set of all induced paths in $G$. We  partition them into \emph{buckets}.
For a pair of distinct vertices $u,v$, the bucket $\Bc_{u,v}$ contains all induced paths with one endvertex $u$ and the other $v$.
Since $G$ is $P_t$-free, the total size of all the buckets is 
$|V(G)|^{t-1}$.

Let $\eps > 0$ be a constant. We say that a vertex $w$ \emph{$\eps$-hits} a bucket $\Bc_{u,v}$ if $N[w]$ intersects at least $\eps \cdot |\Bc_{u,v}|$ paths in $\Bc_{u,v}$. 
A vertex $w$ is \emph{$\eps$-heavy} if it $\eps$-hits at least $\eps \cdot \binom{|V(G)|}{2}$ buckets (i.e., if $N[w]$ intersects at least an $\eps$-fraction of paths in at least an $\eps$-fraction of buckets).
The crucial idea of our algorithm is encapsulated in  the following claim, whose proof is inspired by the result of Gartland and Lokshtanov~\cite{GL20}.

\begin{lemma}\label{lem:ptfree-heavy}
A connected $P_t$-free graph has a $\frac{1}{2t}$-heavy vertex.
\end{lemma}

\begin{proof}
Let $n$ be the number of vertices of the considered graph $G$.
Let $X$ be the set given by \cref{lem:separator} for $A = V(G)$.
We claim that $N[X]$ intersects all paths in at least $\frac{1}{2} \binom{n}{2}$ buckets.

Observe that $\Bc_{u,v}$ is non-empty if and only if $u$ and $v$ are in the same connected component.
So, as $G$ is connected, all the buckets are non-empty. As each connected component in $G - N[X]$ has at most $n/2$ vertices,
the number of buckets that contain at least one path disjoint with $N[X]$ is 
\[
\sum_{C: \text{ component of } G - N[X]} \binom{|V(C)|}{2} \leq 2 \cdot \frac{\frac{n}{2} \left( \frac{n}{2}-1\right)}{2} \leq \frac{1}{2}  \binom{n}{2}.
\]
Here, the first inequality follows from the fact that $|V(C)| \leq n/2$ and the convexity of the mapping $x\mapsto \frac{x(x-1)}{2}$.
It follows that $N[X]$ intersects all the paths in at least half of the buckets.

Recall that $|X| \leq t$. Thus, by the pigeonhole principle, there is $w \in X$ such that $N[w]$ intersects at least $\frac{1}{t}$-fraction of paths in at least $\frac{1}{2t} \binom{n}{2}$ buckets.
\end{proof}

We now proceed to describing the algorithm. For simplicity, the algorithm returns the maximum weight of a solution, but it is straightforward to adapt it so that a solution witnessing this value is constructed on the way. 
The key step is \emph{branching on heavy vertices}. For a vertex~$w$, we will separately consider two instances: $G-w$ (indicating that $w$ is not chosen to the solution, we call this the \emph{failure branch}),
and $G - N[w]$ (indicating that $w$ is chosen to the solution, this branch is called \emph{successful}).
Clearly, the optimum weight of a solution is the maximum of the return value of the first call and the return value of the second call, plus $\wei(w)$.

Now, the algorithm is very simple. If the vertex set is empty, then we return 0.
If $G$ has one vertex, then we return its weight.
If $G$ is disconnected, we call the algorithm recursively for every connected component of $G$.
Otherwise, we enumerate all induced paths in $G$ and partition them into buckets, find a $\frac{1}{2t}$-heavy vertex, and we branch on it.
The pseudo-code is given in \Cref{alg:ptfree}.

\medskip

\begin{algorithm}[H]
\caption {FindMIS \label{alg:ptfree}}
\SetKwFunction{algo}{FindMIS}
\KwIn{$P_t$-free graph $G$, weight function $\wei$}
\If {$|V(G)| \leq 1$} {\Return the total weight of $V(G)$}
\If {$G$ is disconnected \label{line:components}} 
{
\Return{$\sum_{C: \textrm{component of } G} \algo(C,\wei)$ \label{line:callcomponents}} 
}
Initialize buckets $\Bc_{u,v}$ for all $u,v \in V(G)$\\
$w \gets \text{ a } \frac{1}{2t}$-heavy vertex in $G$ \label{line:heavy}\\
\Return {$\max \{ \algo(G-w,\wei), \algo(G-N[w],\wei) + \wei(w) \}$ \label{line:branching} }
\end{algorithm}
\medskip

\cref{lem:ptfree-heavy} asserts that in line~\ref{line:heavy} we can always find a $\frac{1}{2t}$-heavy vertex.
Note that in each recursive call the number of vertices of the instance graph decreases. Thus it is clear that the algorithm terminates and returns the correct value.
Furthermore, the local computation in each node of the recursion tree can be performed in time $|V(G)|^{\Oh(t)}$.
It remains to show that the number of nodes in the recursion tree is bounded by $|V(G)|^{\Oh(\log^2 |V(G)|)}$.

To this end, consider the recursion tree $T$ of the algorithm applied on a graph $G$. For a call on a graph $H$ (which is a subgraph $G$), the \emph{local subtree} of the call
consists of all descendant calls that treat a graph with at least $0.99|V(H)|$ vertices. 
We greedily find a partition $\mathcal{P}$ of $T$ into local subtrees as follows: start with $\mathcal{P} = \emptyset$ and, as long as there exists a call in $T$
that is in none of the local subtrees in $\mathcal{P}$, take such a call closest to the root and add its local subtree to $\mathcal{P}$.

Clearly, a root-to-leaf path in the recursion tree intersects $\Oh(\log |V(G)|)$ local subtrees of $\mathcal{P}$. Thus, it suffices to prove that
any local subtree, say for a call on a graph $H$, contains at most $|V(H)|^{\Oh(\log |V(H)|)}\leq |V(G)|^{\Oh(\log |V(G)|)}$ leaves.

Let $S$ be the local subtree rooted at a call on a graph $H$. Mark the following edges of $S$:
\begin{enumerate}
\item For every call in $S$ on a disconnected graph, say on $H'$,
observe that there is at most one child call of this call that also belongs to $S$. Indeed, this call must be on a connected component $H''$ of $H'$ satisfying $|V(H'')| \geq 0.99|V(H')|$, and there is at most one such component. If there exists such a unique child call that belongs to $S$, mark the edge to it. 
\item For every call in $S$ on a connected graph, mark the edge to the call in the failure branch (provided it belongs to $S$).
\end{enumerate}
Thus, every call in $S$ has at most one marked edge to a child. Hence, the marked edges form a family of vertex-disjoint upwards paths in $S$. 
Let $S'$ be the tree obtained from $S$ by contracting all the marked edges; the parent-child relation is naturally inherited from $S$. Then, every node in $S'$ has $\Oh(|V(H)|)= \Oh(|V(G)|)$ children and every edge of $S'$ corresponds to a successful branch in some call in $S$.
It suffices to show that $S'$ has depth $\Oh(\log |V(H)|)= \Oh(\log |V(G)|)$. 

To this end, we introduce the potential of a call in $S$, say on a graph $H'$:
\[
\mu(H') \coloneqq - \sum_{\{u,v\} \in \binom{V(H')}{2}} \log_{(1-1/2t)} (1+|\Bc_{u,v}|).
\]
At the initial call on $H$, we have $\mu(H) = \Oh(|V(H)|^2 \log |V(H)|)$, because the size of each
bucket is 
at most $|V(H)|^{t-1}$.
Since in a successful branch we remove the closed neighborhood of a $\frac{1}{2t}$-heavy vertex,
a successful branch at a call in $S$ on a graph $H'$ results in decreasing   
 the potential $\mu$ by at least
\[
\frac{1}{2t} \binom{|V(H')|}{2} \geq \frac{1}{2t} \binom{\lceil 0.99|V(H)| \rceil}{2} \geq \frac{0.9}{2t} \binom{|V(H)|}{2}.
\]
Since $\mu$ is nonnegative, it follows that the depth of $S'$ is bounded by $\Oh(\log |V(H)|)$, as desired.

\section{Extensions of the algorithm}\label{sec:cor}
In this section we discuss possible extensions of the algorithm from \cref{sec:ptfree}.
The crucial step of many subexponential-time algorithms for $P_t$-free graphs and $C_{>t}$-free graphs is branching on a high-degree vertex~\cite{BacsoLMPTL19,GORSSS18,ORz20,CPPT}.
We observe that some of these algorithms can be turned into quasi-polynomial ones with the new approach.

\subsection{Partitioning vertices: \Col}

Let us consider the \Col problem, whose complexity in $P_t$-free graphs is a well-known open problem~\cite{BonomoCMSSZ18,DBLP:journals/corr/abs-2006-03009}. We aim to show \cref{thm:threecol}, restated below.

\threecol*

The algorithm and its analysis are very similar to the proof of \cref{thm:main-ptfree}, so we will only point out the differences.
The adaptation is inspired by the known subexponential-time algorithms for \Col~\cite{BRz19,GORSSS18}.
We remark that a quasi-polynomial-time algorithm for \Col in $P_t$-free graphs with running time $n^{\Oh(\log^3 n)}$ can be also derived from the work of Gartland and Lokshtanov~\cite{GL20}, by an analogous adaptation of their approach.

Actually, we will solve the more general \textsc{List 3-Coloring}, where every vertex $v$ has a \emph{list} $L(v) \subseteq \{1,2,3\}$, and we ask for a coloring respecting lists $L$ in the sense that the color of each vertex belongs to its list.

As the first step, we preprocess the instance as follows.
If there exists a vertex with an empty list, then there is no way to properly color the graph with lists $L$ and thus we can immediately terminate the current call, as we deal with a no-instance.
Further, if there is a vertex $v$ with a one-element list, say $L(v) = \{c\}$, then we can obtain an equivalent instance by removing $c$ from the lists of neighbors of $v$ and deleting $v$ from the graph. This corresponds to coloring $v$ with the color $c$. 
Finally, we enumerate all sets $S \subseteq V(G)$ of size at most $t-1$, and all their proper colorings, respecting lists $L$.
If for some set $S$, the graph $G[S]$ cannot be properly colored with lists $L$, then we terminate the call and report a no-instance.
Moreover, if for some $S$, some $v \in S$, and some $c \in L(v)$, the vertex $v$ is not colored $c$ in any proper coloring of $G[S]$, respecting lists $L$, then we can safely remove $c$ from $L(v)$.
We perform these steps exhaustively; this can clearly be done in polynomial time.
Thus, after the preprocessing, the instance satisfies the following properties:
\begin{enumerate}[(P1)]
\item Each list has two or three elements. \label{prop:lists}
\item For each $v \in V(G)$, each $c \in L(v)$, and each $S \subseteq V(G)$, such that $v \in S$ and $|S| \leq t-1$, there is a proper coloring of $G[S]$, respecting lists $L$, in which the color of $v$ is $c$. \label{prop:consistency}
\end{enumerate}

Similarly to \Cref{alg:ptfree}, our algorithm has two key steps. 
If the graph is disconnected, we call the algorithm for each connected component independently, and report a yes-instance if all these calls report yes-instances.

Otherwise, if the graph is connected, then we will branch on a vertex. Again, this vertex will be carefully chosen using buckets.
This time the objects in a bucket $\Bc_{u,v}$ will be \emph{colored} induced $u-v$ paths, i.e., for each induced $u-v$ path we additionally enumerate all its proper colorings, respecting lists $L$.
Observe that by property \ref{prop:consistency} we know that every induced $u-v$ path $P$ appears at least once in $\Bc_{u,v}$ as a colored path. Even stronger, if $w$ is a vertex of $P$ and $c \in L(w)$, then $P$ appears in $\Bc_{u,v}$ as a colored path, where $w$ is colored $c$.
Observe that thus, we still have the property that $\Bc_{u,v}$ is non-empty if and only if $u$ and $v$ are in the same connected component of $G$.
Note that the total size of all buckets is at most $|V(G)|^{t-1} \cdot 3^{t-1}=(3|V(G)|)^{t-1}$, and we can compute them in time $|V(G)|^{\Oh(t)}$.

The crucial observation is that the analogue of \cref{lem:ptfree-heavy} holds for these buckets too.

\begin{lemma}\label{lem:heavy-coloring}
Suppose $G$ is a connected $P_t$-free graph and $L\colon V(G)\to 2^{\{1,2,3\}}$ is a list assignment that satisfies properties~\ref{prop:lists} and \ref{prop:consistency}. Then there is a vertex $w$ and a color $c \in L(w)$
such that for at least $\frac{1}{8t} \cdot \binom{|V(G)|}{2}$ pairs $\{u,v\} \in \binom{V(G)}{2}$,
at least $\frac{1}{8t \cdot 3^{t-1}} \cdot |\Bc_{u,v}|$ colored paths in $\Bc_{u,v}$ contain a vertex colored $c$ that belongs to~$N[w]$.
\end{lemma}
\begin{proof}
For all $\{u,v \} \in \binom{V(G)}{2}$, let $\Bc'_{u,v}$ be the bucket from \Cref{alg:ptfree}, i.e., the collection of all induced $u-v$ paths.
Let $w$ be a $\frac{1}{2t}$-heavy vertex given by \cref{lem:ptfree-heavy} for the buckets $\{\,\Bc'_{u,v}\ \colon\ \{u,v\} \in \binom{V(G)}{2}\,\}$.

Recall that by property \ref{prop:lists}, each vertex in $G$ has one of four possible lists.
Thus, by the pigeonhole principle, there exists a list $R\subseteq \{1,2,3\}$ and a subset $Q \subseteq \binom{V(G)}{2}$ of size at least $\frac{1}{8t} \binom{|V(G)|}{2}$
such that for all $\{u,v\} \in Q$, there exists $\widetilde{\Bc}'_{u,v} \subseteq \Bc'_{u,v}$ of size at least $\frac{1}{8t} \cdot |\Bc'_{u,v}|$ with the property that each path in $\widetilde{\Bc}'_{u,v}$ contains a vertex that belongs to $N[w]$ and whose list is $R$.

By property \ref{prop:lists} we know that $|L(w)|\geq 2$ and $|R|\geq 2$, so there is a color $c \in R \cap L(w)$.
Furthermore, by property \ref{prop:consistency}, each path in $\widetilde{\Bc}'_{u,v}$ gives rise to at least one colored path in $\Bc_{u,v}$ which contains a vertex colored~$c$. Moreover, these colored paths are pairwise distinct.

Now the claim follows from the observation that $|\Bc_{u,v}| \leq 3^{t-1} \cdot |\Bc'_{u,v}|$, so for each $\{u,v\} \in Q$, we have selected at least $|\widetilde{\Bc}'_{u,v}| \geq \frac{1}{8t} |\Bc'_{u,v}| \geq \frac{1}{8t \cdot 3^{t-1}} |\Bc_{u,v}|$ paths in $\Bc_{u,v}$.
\end{proof}

\cref{lem:heavy-coloring} gives us an efficient way to perform branching in case of dealing with a connected graph.
Let $w$ and $c$ be as in the statement of the lemma, note that they can be found in polynomial time, as the total size of all buckets is $|V(G)|^{\Oh(t)}$.
We branch on coloring $w$ with color $c$.
In the first branch we remove $c$ from $L(w)$ (this corresponds to deciding not to color $w$ with $c$).
In the second branch, we assign the color $c$ to $w$, i.e., we remove all other colors from $L(w)$.
Note that the preprocessing step at the beginning of the subsequent recursive call will remove the vertex $w$ from the graph, and the color $c$ from the lists of all the neighbors of $w$.
We report a yes-instance if at least one of the two recursive calls reports a yes-instance.

The algorithm clearly terminates, as in each call we reduce the total size of all lists, and returns the correct answer.
Recall that the size of each bucket is polynomial in $|V(G)|$.
\cref{lem:heavy-coloring} implies that in the preprocessing phase in the branch where we color $w$ with the color $c$, we remove a constant fraction of colored paths in a constant fraction of buckets.
Note that a path may be removed in one of two ways: either it contains $w$, so it will be removed when we delete $w$, or it contains a neighbor of $w$ colored $c$, so it will be removed when we delete $c$ from the lists of neighbors of $w$.

Now the analysis of the running time of the algorithm is essentially the same as that in the proof of \cref{thm:main-ptfree}; we leave the details to the reader. This completes the proof of \cref{thm:threecol} in the case for $P_t$-free graphs.

\bigskip

Let us point out that the above algorithm works also in the weighted setting, i.e., with each pair $(v,c)$, where $v \in V(G)$ and $c \in \{1,2,3\}$, we are given a cost $\wei(v,c)$ of coloring $v$ with $c$, and we ask for a proper coloring minimizing the total cost. A natural special case of this problem is \textsc{Independent Odd Cycle Transversal}, where we ask for a minimum-sized independent set which intersects all odd cycles. The complexity of this problem in $P_t$-free graphs is another open problem in the area~\cite{BonamyDFJP19}.

Furthermore, the algorithm from this section can be extended to some family of (weighted) graph homomorphism problems, which generalize both the \MIS problem and the (\textsc{List}) \Col problem, similarly to the work of Groenland et al.~\cite{GORSSS18}, see also~\cite{ORz20}. We skip the details, as they do not bring any new insight.

\subsection{Packing fixed patterns: \MIM}

Another way to generalize \MIS is to pack induced, non-adjacent copies of some fixed pattern in the host graph $G$. A natural example of such a problem is \MIM, where the pattern is $K_2$. This problem can be equivalently formulated as the \MIS problem on $L^2(G)$, i.e., the square of the line graph of $G$. The vertex set of $L^2(G)$ is $E(G)$, and the edges $e_1,e_2 \in E(G)$ are adjacent in $L^2(G)$ if and only if they do not form an induced matching in $G$, i.e., they either intersect, or some vertex of $e_1$ is adjacent to some vertex of $e_2$.
The following structural property of $L^2(G)$ was shown by Kobler and Rotics~\cite{KR}.

\begin{lemma}[Kobler and Rotics~\cite{KR}] \label{lem:linegraph}
For every $t \geq 4$, if $G$ is $P_t$-free, then $L^2(G)$ is $P_t$-free.
\end{lemma}

As the number of vertices of $L^2(G)$ is at most $|V(G)|^2$, \cref{lem:linegraph} combined with \cref{thm:main-ptfree} immediately yields the following.

\begin{corollary}
For every fixed $t\in \mathbb{N}$, the \MIM problem can be solved in time $n^{\Oh(\log^2 n)}$ in $n$-vertex $P_t$-free graphs.
\end{corollary}

Let us point out that we could also obtain an algorithm for \MIM by a direct modification of \Cref{alg:ptfree} in a spirit similar to \cref{thm:threecol}, see also~\cite{BRz19}. Moreover, the algorithm can be further generalized to solve  the \textsc{Maximum $\mathcal{H}$-Packing} problem for any fixed family $\mathcal{H}$ of graphs. In this problem we ask for a maximum-size (or, more generally, maximum-weight) set $X$ such that every connected component of $G[X]$ is isomorphic to some graph in $\mathcal{H}$. Again, we skip the technical details and refer the reader to~\cite{BRz19}.

\subsection{Finding induced subgraphs of bounded treewidth: \textsc{Min Feedback Vertex Set}}

Another way to look at \MIS is to find a maximum-weight induced subgraph of treewidth 0.
A natural next step is to look for a maximum induced forest, i.e., a subgraph of treewidth 1.
By complementation, this problem is equivalent to the \textsc{Min Feedback Vertex Set} problem, where we want to find a minimum-size (or minimum-weight) set which intersects all cycles.

A subexponential-time algorithm for \textsc{Min Feedback Vertex Set} in $P_t$-free graphs is known~\cite{NovotnaOPRLW19} and also involves branching on a high-degree vertex.
However, it has also one more step of exhaustive guessing the large-degree vertices that are not in the optimum feedback vertex set, and it is not clear how to avoid this.
It would be interesting to investigate if the methods of Gartland of Lokshtanov~\cite{GL20} or our approach could be used to obtain a quasi-polynomial-time algorithm for \textsc{Min Feedback Vertex Set} in $P_t$-free.


\bibliographystyle{abbrv}
\bibliography{references}

\begin{thebibliography}{10}

\bibitem{Alekseev82}
V.~Alekseev.
\newblock The effect of local constraints on the complexity of determination of
  the graph independence number.
\newblock {\em Combinatorial-algebraic methods in applied mathematics}, pages
  3--13, 1982.
\newblock (in Russian).

\bibitem{Al04}
V.~E. Alekseev.
\newblock Polynomial algorithm for finding the largest independent sets in
  graphs without forks.
\newblock {\em Discrete Applied Mathematics}, 135(1--3):3--16, 2004.

\bibitem{BacsoLMPTL19}
G.~Bacs{\'{o}}, D.~Lokshtanov, D.~Marx, M.~Pilipczuk, Z.~Tuza, and E.~J. van
  Leeuwen.
\newblock Subexponential-time algorithms for maximum independent set in
  ${P}_t$-free and broom-free graphs.
\newblock {\em Algorithmica}, 81(2):421--438, 2019.

\bibitem{BonamyDFJP19}
M.~Bonamy, K.~K. Dabrowski, C.~Feghali, M.~Johnson, and D.~Paulusma.
\newblock Independent feedback vertex set for ${P}_5$-free graphs.
\newblock {\em Algorithmica}, 81(4):1342--1369, 2019.

\bibitem{BRz19}
{\'{E}}.~Bonnet and P.~Rz\k{a}\.zewski.
\newblock Optimality program in segment and string graphs.
\newblock {\em Algorithmica}, 81(7):3047--3073, 2019.

\bibitem{BonomoCMSSZ18}
F.~Bonomo, M.~Chudnovsky, P.~Maceli, O.~Schaudt, M.~Stein, and M.~Zhong.
\newblock Three-coloring and list three-coloring of graphs without induced
  paths on seven vertices.
\newblock {\em Comb.}, 38(4):779--801, 2018.

\bibitem{CPPT}
M.~Chudnovsky, M.~Pilipczuk, M.~Pilipczuk, and S.~Thomass{\'{e}}.
\newblock On the maximum weight independent set problem in graphs without
  induced cycles of length at least five.
\newblock {\em {SIAM} J. Discret. Math.}, 34(2):1472--1483, 2020.

\bibitem{ChudnovskyPPT20}
M.~Chudnovsky, M.~Pilipczuk, M.~Pilipczuk, and S.~Thomass{\'{e}}.
\newblock Quasi-polynomial time approximation schemes for the maximum weight
  independent set problem in \emph{H}-free graphs.
\newblock In S.~Chawla, editor, {\em Proceedings of the 2020 {ACM-SIAM}
  Symposium on Discrete Algorithms, {SODA} 2020, Salt Lake City, UT, USA,
  January 5-8, 2020}, pages 2260--2278. {SIAM}, 2020.

\bibitem{DBLP:journals/corr/abs-2006-03009}
M.~Chudnovsky, S.~Spirkl, and M.~Zhong.
\newblock List 3-coloring {$P_t$}-free graphs with no induced 1-subdivision of
  {$K_{1,s}$}.
\newblock {\em Discrete Mathematics}, 343(11):112086, 2020.

\bibitem{CorneilLB81}
D.~G. Corneil, H.~Lerchs, and L.~S. Burlingham.
\newblock Complement reducible graphs.
\newblock {\em Discrete Applied Mathematics}, 3(3):163--174, 1981.

\bibitem{GL20}
P.~Gartland and D.~Lokshtanov.
\newblock Independent set on ${P}_k$-free graphs in quasi-polynomial time.
\newblock {\em CoRR}, abs/2005.00690, 2020.
\newblock Accepted to FOCS 2020.

\bibitem{GORSSS18}
C.~Groenland, K.~Okrasa, P.~Rz\k{a}\.{z}ewski, A.~Scott, P.~Seymour, and
  S.~Spirkl.
\newblock ${H}$-colouring ${P}_t$-free graphs in subexponential time.
\newblock {\em Discret. Appl. Math.}, 267:184--189, 2019.

\bibitem{GrzesikKPP19}
A.~Grzesik, T.~Klimo\v{s}ov\'{a}, M.~Pilipczuk, and M.~Pilipczuk.
\newblock Polynomial-time algorithm for maximum weight independent set on
  ${P}_6$-free graphs.
\newblock In T.~M. Chan, editor, {\em Proceedings of the Thirtieth Annual
  {ACM-SIAM} Symposium on Discrete Algorithms, {SODA} 2019, San Diego,
  California, USA, January 6-9, 2019}, pages 1257--1271. {SIAM}, 2019.

\bibitem{gyarfas}
A.~Gy\'arf\'as.
\newblock Problems from the world surrounding perfect graphs.
\newblock {\em Applicationes Mathematicae}, 19:413--441, 1987.

\bibitem{KR}
D.~Kobler and U.~Rotics.
\newblock Finding maximum induced matchings in subclasses of claw-free and
  {$P_5$}-free graphs, and in graphs with matching and induced matching of
  equal maximum size.
\newblock {\em Algorithmica}, 37(4):327–346, Dec. 2003.

\bibitem{LokshtanovVV14}
D.~Lokshtanov, M.~Vatshelle, and Y.~Villanger.
\newblock Independent set in ${P}_5$-free graphs in polynomial time.
\newblock In C.~Chekuri, editor, {\em Proceedings of the Twenty-Fifth Annual
  {ACM-SIAM} Symposium on Discrete Algorithms, {SODA} 2014, Portland, Oregon,
  USA, January 5-7, 2014}, pages 570--581. {SIAM}, 2014.

\bibitem{LozinM08}
V.~V. Lozin and M.~Milanic.
\newblock A polynomial algorithm to find an independent set of maximum weight
  in a fork-free graph.
\newblock {\em J. Discrete Algorithms}, 6(4):595--604, 2008.

\bibitem{Minty80}
G.~J. Minty.
\newblock On maximal independent sets of vertices in claw-free graphs.
\newblock {\em J. Comb. Theory, Ser. {B}}, 28(3):284--304, 1980.

\bibitem{NovotnaOPRLW19}
J.~Novotn{\'{a}}, K.~Okrasa, M.~Pilipczuk, P.~Rz\k{a}\.zewski, E.~J. van
  Leeuwen, and B.~Walczak.
\newblock Subexponential-time algorithms for finding large induced sparse
  subgraphs.
\newblock {\em Algorithmica}, July 2020.

\bibitem{ORz20}
K.~Okrasa and P.~Rz\k{a}\.{z}ewski.
\newblock Subexponential algorithms for variants of the homomorphism problem in
  string graphs.
\newblock {\em J. Comput. Syst. Sci.}, 109:126--144, 2020.

\bibitem{Sbihi80}
N.~Sbihi.
\newblock Algorithme de recherche d'un stable de cardinalite maximum dans un
  graphe sans etoile.
\newblock {\em Discrete Mathematics}, 29(1):53--76, 1980.
\newblock (in French).

\end{thebibliography}
\end{document}